\documentclass[envcountsame,final]{svjour3}

\newcommand{\ket}[1]{|#1\rangle}
\newcommand{\bra}[1]{\langle #1 |}

\usepackage[T1]{fontenc}
\usepackage{txfonts}

\begin{document}
\hypersetup{pdftitle={Coding Theoretic Construction of Quantum Ramp Secret Sharing},pdfauthor={Ryutaroh Matsumoto},pdfkeywords={algebraic geometry code, non-perfect secret sharing, quantum secret sharing, ramp secret sharing}}
\title{Coding Theoretic Construction of Quantum Ramp Secret Sharing}
\author{Ryutaroh Matsumoto}
\date{December 2014}
\institute{Ryutaroh Matsumoto \at Department of Communications and Computer Engineering, Tokyo Institute of Technology, 152-8550 Japan\\
and Department of Mathematical Sciences,
Aalborg University, Denmark\\
              ORCID: 0000-0002-5085-8879 \\
              \email{ryutaroh@it.ce.titech.ac.jp}}
\maketitle
\begin{abstract}
We show a construction of a quantum ramp secret sharing scheme
from a nested pair of linear codes.
Necessary and sufficient conditions for qualified sets and forbidden sets
are given in terms of combinatorial properties of nested linear codes.
An algebraic geometric construction for quantum secret sharing
is also given.
\end{abstract}

\section{Introduction}
Secret sharing (SS) \cite{shamir79} is a cryptographic scheme to
encode a secret to multiple shares being distributed to
participants, so that only qualified (or authorized)
sets of participants
can reconstruct the original secret from their shares.
Traditionally both secret and shares were classical information
(bits). Several authors \cite{cleve99,gottesman00,smith00}
extended the traditional SS to quantum one
so that a quantum secret can be encoded to quantum shares.

When we require unqualified sets of participants
to have zero information of the secret,
the size of each share must be larger than or equal to
that of secret.
By tolerating partial information leakage to
unqualified sets, the size of shares can be smaller
than that of secret. Such an SS is called a ramp (or non-perfect) SS
\cite{blakley85,ogata93,yamamoto86}. The quantum ramp SS was
proposed by Ogawa et al.\ \cite{ogawa05}.
In their construction \cite{ogawa05} as well as
its improvement \cite{matsumoto14strong},
the size of shares can be $L$ times smaller relative to quantum secret
than its previous construction \cite{cleve99,gottesman00,smith00},
where $L$ is the number of qudits in quantum secret.

In their construction \cite{ogawa05}, each share is a quantum state
on a $q$-dimensional complex linear space,
and $q$ has to be larger than or equal to the number $n$
of participants.
When $n$ is large, $q$ also has to be large.
But it is not clear whether or not such a large
dimensional quantum systems are always readily available.
To deal with such a situation, we need a quantum ramp SS
allowing $n > q$.
We stress that we study the ramp (non-perfect) SS
while \cite{cleve99,gottesman00,smith00} and their
subsequent developments \cite{marin13,markham08} studied the perfect
SS, and that none of the results in this paper are
contained in \cite{cleve99,gottesman00,markham08,sarvepalli12,smith00}.

On the other hand,
the present paper can be regarded as a generalization of
\cite{gottesman00,sarvepalli12}.
Because \cite{gottesman00,sarvepalli12} studied
connection between perfect quantum SS and
the Calderbank-Shor-Steane (CSS) quantum error-correcting codes
\cite{calderbank96,steane96},
while our proposed encoding (\ref{eq1}) of quantum secret
into quantum shares is the same as
that of the $q$-ary CSS codes.
The connection between quantum \emph{ramp} SS and quantum error
correction seems first studied in \cite{marin13}.
Our new contributions that are not given in \cite{marin13}
are (a) necessary and sufficient conditions for qualified sets
and forbidden sets
that can be easily checked by a digital computer,
(b) a quantum procedure partially reconstructing the quantum secret
by an intermediate set of shares,
and (c) a construction of quantum ramp SS that allows arbitrarily large
$n$ for a fixed $q$.
Item (a) completely characterizes the qualified and the forbidden
sets. Such a complete characterization
cannot be obtained by regarding the reconstruction
of quantum secret as the erasure decoding of quantum error-correcting
codes \cite{marin13}.
Item (b) above clarifies how much quantum information in the secret
can
be reconstructed by an intermediate set, which is a share set
neither
qualified nor forbidden (unauthorized).
We note that item (c) above does not contradict with
$q > \sqrt{(n+2)/2}$ \cite[Eq.\ (5)]{marin13}, because
\cite[Eq.\ (5)]{marin13} considered perfect quantum SS.

It is well-known that all linear classical ramp SS can be constructed from
a pair of linear codes $C_2 \subsetneq C_1 \subseteq \mathbf{F}_q^n$
\cite{chen07,cruz10},
where $\mathbf{F}_q$ is the finite field with $q$
elements.
Smith \cite{smith00} studied connection between
\emph{perfect} linear classical SS and \emph{perfect} quantum SS
by using the monotone span program that can
express any \emph{perfect} linear classical SS, but
he did not considered ramp SS.
We call a quantum state in a $q$-dimensional
system as a qudit.
In this paper we shall show the following.
\begin{theorem}\label{thm1}
Let $J \subseteq \{1$, \ldots, $n\}$
and $\overline{J} = \{1$, \ldots, $n\} \setminus J$.
For $\vec{x} = (x_1$, \ldots, $x_n) \in \mathbf{F}_q^n$
define $P_J(\vec{x}) =(x_i)_{i \in J}$.
We define $\widetilde{P}_J$ to be an
$\mathbf{F}_q$-linear map from $C_1/C_2$ to $P_J(C_1)/P_J(C_2)$ sending
$\vec{x} + C_2 \in C_1/C_2$ to $P_J(\vec{x}) + P_J(C_2) \in P_J(C_1)/P_J(C_2)$.
A quantum ramp SS can be constructed from \textbf{any}
$C_2 \subsetneq C_1 \subseteq \mathbf{F}_q^n$,
regardless of $n$ and $q$.
\begin{enumerate}
\item\label{l:encode} The constructed quantum SS encodes
a quantum secret of $(\dim C_1 - \dim C_2)$
qudits to $n$ shares. Each share is a qudit.
\item\label{l:decode}
A set $J$ of participants can reconstruct
\begin{equation}
\dim \widetilde{P}_J(\ker(\widetilde{P}_{\overline{J}}))
\label{eq401}
\end{equation}
qudits out of $(\dim C_1 - \dim C_2)$
qudits of the encoded quantum secret.
If 
\begin{equation}
\dim \widetilde{P}_J(\ker(\widetilde{P}_{\overline{J}}))
= \dim C_1 - \dim C_2
\label{eq402}
\end{equation}
then the set $J$ of participants can reconstruct the  secret
perfectly. This means that $J$ is a qualified set.
In this case $\overline{J}$ has no information of the secret,
which means that $\overline{J}$ is a forbidden (also called
unauthorized)
set.
\item\label{l:implication}
The condition (\ref{eq402}) is equivalent to both
\begin{eqnarray}
\dim P_J(C_1)-\dim P_J(C_2) &=& \dim C_1-\dim C_2 \textrm{ and}\label{eq101}\\ 
\dim P_{\overline{J}}(C_1) - \dim P_{\overline{J}}(C_2)&=& 0. \label{eq102}
\end{eqnarray}
Condition (\ref{eq102}) is equivalent to
\begin{equation}
\dim C_2^\perp \cap \ker(P_J) - \dim C_1^\perp \cap \ker(P_J) = 0. \label{eq403}
\end{equation}
\item\label{l:necessity} Both (\ref{eq101}) and (\ref{eq102})
are also a necessary condition for $J$ to be a qualified set.
\end{enumerate}
\end{theorem}

This paper is organized as follows:
Section \ref{sec2} proposes the encoding of secrets
and shows Item \ref{l:encode} in Theorem \ref{thm1}.
Section \ref{sec3} proposes the decoding of secrets and
it shows Items \ref{l:decode}
and \ref{l:implication} in Theorem \ref{thm1}.
Section \ref{sec11} proves Item \ref{l:necessity} in
Theorem \ref{thm1} by computing the Holevo information of
the set $J$. It also computes the coherent information as a byproduct.
Section \ref{sec4} shows that Theorem \ref{thm1}
completely characterizes the qualified and forbidden
sets of the quantum ramp SS by
Ogawa et al.\ \cite{ogawa05}.
Section \ref{sec5} gives an algebraic geometric (AG) construction.
A major benefit of the AG construction is that $n$ can become arbitrarily
large for a fixed $q$ \cite{bn:stichtenoth}.
Section \ref{sec6} gives concluding discussions.

\section{Encoding Secrets}\label{sec2}
We shall propose a construction of a quantum ramp SS
from a nested pair of linear codes
$C_2 \subsetneq C_1 \subseteq \mathbf{F}_q^n$.
Our proposal is a quantum
version of classical ramp SS
proposed by Chen et al.\ \cite[Section 4.2]{chen07}.
Let $\mathcal{G}_i$ and $\mathcal{H}_j$ be $q$-dimensional complex linear spaces.
We also assume that orthonormal bases of $\mathcal{G}_i$
and $\mathcal{H}_j$ are indexed by $\mathbf{F}_q$
as $\{\ket{s}\}_{s \in \mathbf{F}_q}$.
The quantum secret is $\dim C_1 - \dim C_2$ qudits on
$\bigotimes_{i=1}^{\dim C_1 - \dim C_2} \mathcal{G}_i$.
Fix an $\mathbf{F}_q$-linear isomorphism $f: \mathbf{F}_q^{\dim C_1 - \dim C_2}
\rightarrow C_1 / C_2$.
Also,
$\{ \ket{\vec{s}} \mid \vec{s} \in \mathbf{F}_q^{\dim C_1 - \dim C_2}\}$
is an orthonormal basis of $\bigotimes_{i=1}^{\dim C_1 - \dim C_2} \mathcal{G}_i$.
We shall encode a quantum secret to
$n$ qudits in $\bigotimes_{j=1}^n \mathcal{H}_j$
by a complex linear isometric embedding.
To specify such an embedding, it is enough to specify
the image of each basis state $\ket{\vec{s}} \in \bigotimes_{i=1}^{\dim C_1 - \dim C_2} \mathcal{G}_i$.
We encode $\ket{\vec{s}}$ to
\begin{equation}
\frac{1}{\sqrt{|C_2|}} \sum_{\vec{x} \in f(\vec{s})} \ket{\vec{x}} \in
\bigotimes_{j=1}^n \mathcal{H}_j. \label{eq1}
\end{equation}
We note that the proposed encoding (\ref{eq1}) is equivalent to
that of CSS codes \cite{calderbank96,steane96}.
Recall that by definition of $f$, $f(\vec{s})$ is a subset of $C_1$,
$f(\vec{s}) \cap f(\vec{s}_1) = \emptyset$ if $\vec{s} \neq \vec{s}_1$,
and $f(\vec{s})$ contains $|C_2|$ vectors. From these properties
we see that (\ref{eq1}) defines a complex linear isometric embedding.
The quantum system $\mathcal{H}_j$ is distributed to the $j$-th
participant.

\begin{example}\label{ex1}
We show a slightly modified variant of Ogawa et al.\ \cite{ogawa05}
as an example.
Let $q=7$, $n=5$, $L=3$,
$\alpha_1 = 3$, $\alpha_2 = 5$, $\alpha_3 = 6$,
$\alpha_4 = 1$, $\alpha_5=4$.
For $s_1$, $s_2$, $s_3 \in \mathbf{F}_7$,
$\ket{s_1 s_2s_3}$ is encoded to
\begin{equation}
\frac{1}{\sqrt{7}}\sum_{r\in \mathbf{F}_7}
\bigotimes_{j=1}^5 \ket{r + s_1 \alpha_j + s_2 \alpha_j^2 + s_3 \alpha_j^3}.\label{eq:ex1}
\end{equation}
This encoding can be described by
\begin{eqnarray*}
C_1 &=& \{ (r + s_1 \alpha_j + s_2 \alpha_j^2 + s_3 \alpha_j^3)_{j=1,\ldots, 5} \mid
r, s_1, s_2, s_3 \in \mathbf{F}_7\},\\
C_2 &=& \{ (r,r,r,r,r ) \mid
r \in \mathbf{F}_7\},\\
f(s_1, s_2, s_3) &=& \{ (r + s_1 \alpha_j + s_2 \alpha_j^2+ s_3 \alpha_j^3)_{j=1,\ldots, 5} \mid
r \in \mathbf{F}_7\}.
\end{eqnarray*}
\end{example}

\section{Decoding Secrets}\label{sec3}
\subsection{Preliminary Algebra}
In this subsection we show Item \ref{l:implication} in Theorem
\ref{thm1} in order to introduce the proposed decoding procedure.
The equivalence between (\ref{eq102}) and (\ref{eq403}) follows from
Forney's second duality lemma \cite[Lemma 7]{forney94}
and $\ker(P_J) = \{ (x_1$, \ldots, $x_n) \in \mathbf{F}_q^n \mid
x_i = 0$ if $i\in J \}$.

Equation (\ref{eq101}) is equivalent to $\widetilde{P}_J$ being an
isomorphism, and (\ref{eq102}) is equivalent to $\widetilde{P}_{\overline{J}}$
being the zero map. From
these observations we see that (\ref{eq101}) and (\ref{eq102}) imply
(\ref{eq402}) and vice versa.
This finishes the proof of Item \ref{l:implication} in Theorem \ref{thm1}.

\begin{remark}
Equation (\ref{eq403}) corresponds to \cite[Eq.\ (3)]{kurihara12}
for classical ramp SS.
\end{remark}

\subsection{Proposed Decoding Procedure}
Suppose that the quantum secret is
\begin{equation}
\sum_{\vec{s} \in \mathbf{F}_q^{\dim C_1 -\dim C_2}} \alpha(\vec{s})\ket{\vec{s}}
\in \bigotimes_{i=1}^{\dim C_1 - \dim C_2} \mathcal{G}_i.
\label{eq3}
\end{equation}
It is encoded to $n$ qudits as
\begin{equation}
\sum_{\vec{s} \in \mathbf{F}_q^{\dim C_1 -\dim C_2}} \alpha(\vec{s})
\frac{1}{\sqrt{|C_2|}} \sum_{\vec{x} \in f(\vec{s})} \ket{\vec{x}} \in
\bigotimes_{j=1}^n \mathcal{H}_j. \label{eq4}
\end{equation}
Decompose $\ker(\widetilde{P}_{\overline{J}})$ to a direct sum
$V \oplus (\ker(\widetilde{P}_{\overline{J}}) \cap \ker(\widetilde{P}_J))$,
and decompose
$C_1/C_2$ to $W \oplus V \oplus \cap \ker(\widetilde{P}_J)$.
Let $\mathcal{G}(J)$ to be the complex linear space spanned by
$\{\ket{\vec{s}} \mid f(\vec{s}) \in V\}$.
We have $\dim \mathcal{G}(J) = |\widetilde{P}_J(\ker(\widetilde{P}_{\overline{J}}))|$ because 
\begin{eqnarray}
&&\dim \widetilde{P}_J(\ker(\widetilde{P}_{\overline{J}}))\nonumber\\
&=& \dim \ker(\widetilde{P}_{\overline{J}}) - \dim \ker(\widetilde{P}_{\overline{J}}) \cap \ker(\widetilde{P}_J)\nonumber\\
&=&\dim V. \label{eq:vdim}
\end{eqnarray}
The space $\bigotimes_{i=1}^{\dim C_1 - \dim C_2} \mathcal{G}_i$
can be decomposed as $\mathcal{G}(J) \otimes \mathcal{G}_{\mathrm{rest}}$,
where $\mathcal{G}_{\mathrm{rest}}$ is the complex linear space
spanned by $\{ \ket{\vec{s}_{KW}} \mid f(\vec{s}_{KW}) \in W \oplus \ker(\widetilde{P}_J) \}$,
and $\ket{\vec{s}_J} \otimes \ket{\vec{s}_W+\vec{s}_K}
\in \mathcal{G}(J) \otimes \mathcal{G}_{\mathrm{rest}}$
is identified with $\ket{\vec{s}} \in \bigotimes_{i=1}^{\dim C_1 - \dim C_2} \mathcal{G}_i$ for $\vec{s} = \vec{s}_J + \vec{s}_W+\vec{s}_K$
with $\vec{s}_J \in f^{-1}(V)$, $\vec{s}_W \in f^{-1}(W)$
and $\vec{s}_K \in f^{-1}(\ker(\widetilde{P}_J))$.
This identification is a unitary map between 
$\mathcal{G}(J) \otimes \mathcal{G}_{\mathrm{rest}}$
and $\bigotimes_{i=1}^{\dim C_1 - \dim C_2} \mathcal{G}_i$,
because it is linear and preserves the inner product.

\begin{example}\label{ex2}
We retain the notations from Example \ref{ex1}.
Let $J=\{1,2,3\}$ and $\overline{J} = \{4,5\}$.
Firstly we examine $\ker(\widetilde{P}_{\overline{J}})
\subset C_1 / C_2$.
When $(s_1, s_2, s_3 ) = (2,1,0)$
or $(s_1, s_2, s_3 ) = (0,0,1)$,
$P_{\overline{J}}(f(s_1, s_2,s_3)) = P_{\overline{J}}(C_2)$,
from which we see that
$\ker(\widetilde{P}_{\overline{J}})$ is
two-dimensional linear space spanned by
$f(2, 1, 0)$ and $f(0, 0, 1)$.
On the other hand,
$P_{J}(f(2, 1, 0)) \neq P_{J}(C_2)$ and
$P_{J}(f(0, 0, 1)) = P_{J}(C_2)$,
which mean that
$\ker(\widetilde{P}_{\overline{J}}) \cap \ker(\widetilde{P}_J)$
is
one-dimensional linear space spanned by $f(0, 0, 1)$.
We also observe that $V$ is the one-dimensional space 
spanned by $f(2,1, 0)$,
that $\ker(\widetilde{P}_J)$ is the one-dimensional space 
spanned by $f(0,0,1)$.
There is some freedom in choosing $W$,
for example, we can choose $W$ as
the one-dimensional space 
spanned by $f(1,0,0)$.

$\mathcal{G}(J)$ is the $7$-dimensional complex linear space spanned by
$\{ \ket{2a} \otimes \ket{a} \otimes \ket{0} \mid a \in \mathbf{F}_7 \}$,
while 
$\mathcal{G}_{\mathrm{rest}}$ is 
the $49$-dimensional complex linear space spanned by
$\{ \ket{s_1} \otimes \ket{0} \otimes \ket{s_3}\mid s_1, s_3 \in \mathbf{F}_7 \}$.
\end{example}

In this section we shall prove that
a set $J$ of  participants can reconstruct the part of
the quantum secret (\ref{eq3}) from (\ref{eq4}).
The reconstructed part 
is a state in $\mathcal{G}(J)$.
By reordering indices we may assume $J = \{1$, \ldots, $|J|\}$.
We also assume 
\begin{equation}
\dim \widetilde{P}_J(\ker(\widetilde{P}_{\overline{J}})) > 0, \label{eq:cjpositive}
\end{equation}
otherwise the set $J$ can reconstruct no part of the secret
by the proposed decoding procedure.

The restriction of $\widetilde{P}_J \circ f$ to $V$ is injective by the definition
of $V$.
This and the definitions of $V$ and $W$ imply
 that there exists an $\mathbf{F}_q$-linear isomorphism $g_1$
from $P_J(C_1)/P_J(C_2)$ to $\mathbf{F}_q^{\dim P_J(C_1)-\dim P_J(C_2)}$
with the following condition.
When we write $\vec{s} = \vec{s}_J + \vec{s}_W + \vec{s}_K$ in the same
way as the previous paragraph
for $\vec{s} \in \mathbf{F}_q^{\dim C_1 - \dim C_2}$
then $g_1(\widehat{P}_J(f(\vec{s})) = (\vec{s}_J$, $\vec{s}_W) \in
\mathbf{F}_q^{\dim P_J(C_1)-\dim P_J(C_2)}$.
If (\ref{eq402}) holds then we have
$ V = C_1 / C_2$ and we regard $\vec{s}_W$ and $\vec{s}_K$ as $\vec{0}$ and
$\vec{s}_J$ as $\vec{s}$.
Observe that $g_1$ is inverting the restriction of
$\widetilde{P}_J \circ f$ to $V$.

On the other hand, there also exists an $\mathbf{F}_q$-linear
epimorphism $g_2$ from $P_J(C_1)$ to $\mathbf{F}_q^{\dim P_J(C_2 \cap \ker(P_{\overline{J}}))}$
that is one-to-one on every coset 
belonging to the factor linear space
$P_J(C_1)/P_J(C_2 \cap \ker(P_{\overline{J}}))$.
The above map can be constructed as follows:
Find a direct sum decomposition of
$P_J(C_1) = P_J(C_2 \cap \ker(P_{\overline{J}})) \oplus U$
For $\vec{x}\in P_J(C_1)$,
find a decomposition
$\vec{x} = \vec{x}_1 + \vec{x}_2$ such that
$\vec{x}_1 \in P_J(C_2 \cap \ker(P_{\overline{J}}))$
and $\vec{x}_2 \in U$.
Then map $\vec{x}_1$ by a some fixed linear isomorphism
from $P_J(C_2 \cap \ker(P_{\overline{J}}))$
to $\mathbf{F}_q^{\dim P_J(C_2 \cap \ker(P_{\overline{J}}))}$,
while ignoring $\vec{x}_2$.
Observe that $g_2$ is extracting
the $P_J(C_2 \cap \ker(P_{\overline{J}}))$-component.

By a construction similar to $g_2$,
there also exists an $\mathbf{F}_q$-linear
epimorphism $g_3$ from $P_J(C_1)/P_J(C_2 \cap \ker(P_{\overline{J}}))$ 
to $\mathbf{F}_q^{\dim P_J(C_2) - \dim P_J(C_2 \cap \ker(P_{\overline{J}}))}$
that is one-to-one on  on every coset 
belonging to the factor linear space
$P_J(C_1)/P_J(C_2)$ such that
the value of $g_3$ is determined by $\vec{s}_W$, $\vec{s}_K$,
and $P_{\overline{J}}(\vec{x})$ independently of $\vec{s}_J$.
Observe also that
$g_3$ is extracting the $P_J(C_2)$-component
from the factor linear space 
$P_J(C_1)/P_J(C_2 \cap \ker(P_{\overline{J}}))$.

Consider the $\mathbf{F}_q$-linear map $g_4$ from
$P_J(C_1)$ to $\mathbf{F}_q^{\dim P_J(C_1)}$
sending $\vec{v} \in P_J(C_1)$ to $(g_1(\vec{v}+P_J(C_2))$, $g_2(\vec{v})$,
$g_3(\vec{v}+P_J(C_2 \cap \ker(P_{\overline{J}}))))$.
We see that $g_4$ is an $\mathbf{F}_q$-linear isomorphism
because it is surjective and the domain and the image of
$g_4$ have the same dimension.

For $\vec{v} \in P_J(C_1)$,
we can construct a unitary operation
sending $\ket{\vec{v}} \in \bigotimes_{j=1}^{|J|} \mathcal{H}_j$
to $\ket{g_4(\vec{v}),\vec{0}} \in \bigotimes_{j=1}^{|J|} \mathcal{H}_j$,
where $\vec{0}$ is the zero vector of
length $|J|-\dim P_J(C_1)$.
Since this unitary operation does not change
$\mathcal{H}_{|J|+1}$, \ldots, $\mathcal{H}_n$,
it can be executed only by
the first to the $|J|$-th participants.
Applying the unitary operation to (\ref{eq4}) gives
\begin{eqnarray}
&&\sum_{\vec{s} \in \mathbf{F}_q^{\dim C_1 -\dim C_2}} \alpha(\vec{s})
\frac{1}{\sqrt{|C_2|}} \sum_{\vec{x} \in f(\vec{s})} |
\vec{s}_J, 
\vec{s}_W,\nonumber\\
&& g_2(P_J(\vec{x})), g_3(P_J(\vec{x})+P_J(C_2 \cap \ker(P_{\overline{J}}))),
\vec{0}, P_{\overline{J}}(\vec{x})\rangle.\label{eq6}
\end{eqnarray}
$g_2(P_J(\vec{x}))$ can become any vector
in $\mathbf{F}_q^{\dim P_J(C_2 \cap \ker(P_{\overline{J}}))}$
independently of  $\vec{s}_J$, $\vec{s}_W$, $\vec{s}_K$ and $P_{\overline{J}}(\vec{x})$.
Hereafter we denote $g_2(P_J(\vec{x}))$ by $\vec{u}_1$.
For a fixed $\vec{s} \in \mathbf{F}_q^{\dim C_1 - \dim C_2}$
$P_{\overline{J}}(\vec{x})$ can become any vector in the coset
$\widetilde{P}_{\overline{J}}(f(\vec{s})) \in
P_{\overline{J}}(C_1)/P_{\overline{J}}(C_2)$, and
 $\vec{s}_W$ determines
which coset of $P_{\overline{J}}(C_1)/P_{\overline{J}}(C_2)$
contains $P_{\overline{J}}(\vec{x})$
independently of $\vec{s}_J$, $\vec{s}_K$ and $\vec{u}_1$.
Hereafter we denote the coset $\widetilde{P}_{\overline{J}}(f(\vec{s}))
= P_{\overline{J}}(\vec{x}) + P_{\overline{J}}(C_2)$
by $g_5(\vec{s}_W)$.
By the definition of $g_3$, $g_3(P_J(\vec{x})+P_J(C_2 \cap \ker(P_{\overline{J}})))$ is determined by only $\vec{s}_W$, $\vec{s}_K$ and
$P_{\overline{J}}(\vec{x})$,
that is, independent of $\vec{s}_J$. Hereafter we denote $g_3(P_J(\vec{x})+P_J(C_2 \cap \ker(P_{\overline{J}})))$ by $g_6(\vec{s}_W$, 
$\vec{s}_K$, $P_{\overline{J}}(\vec{x}))$.
By using these notations we can rewrite (\ref{eq6}) as
\begin{equation}
\sum_{\vec{s} \in \mathbf{F}_q^{\dim C_1 -\dim C_2}} \alpha(\vec{s})\ket{\vec{s}_J}
\frac{1}{\sqrt{|C_2|}} \sum_{\begin{array}{l} \scriptstyle \vec{u}_1 \in \mathbf{F}_q^{\dim P_J(C_2 \cap \ker(P_{\overline{J}}))}\\\scriptstyle \vec{u}_2 \in g_5(\vec{s}_W) \end{array}} |
\vec{s}_W, \vec{u}_1, g_6(\vec{s}_W, \vec{s}_K, \vec{u}_2),
\vec{0}, \vec{u}_2\rangle,\label{eq7}
\end{equation}
which means that the part $\ket{\vec{s}_J}$ of
the quantum secret (\ref{eq3}) is reconstructed 
but in general entangled with the rest of quantum system.

If the quantum secret is a product state written as
\[
\sum_{\vec{s} \in \mathbf{F}_q^{\dim C_1-\dim C_2}}
\alpha(\vec{s}) \ket{\vec{s}}
= \left( \sum_{\vec{s}_J \in V}
\alpha(\vec{s}_J) \ket{\vec{s}_J} \right) \otimes
\left(\sum_{\vec{s}_W,\vec{s}_K}
\alpha(\vec{s}_W,\vec{s}_K) \ket{\vec{s}_W,\vec{s}_K}\right)
\]
then (\ref{eq7}) can be written as
\[\left(\sum_{\vec{s}_J \in V}
\alpha(\vec{s}_J)\ket{\vec{s}_J}\right) \otimes
\left( \sum_{\vec{s}_W,\vec{s}_K}
\alpha(\vec{s}_W,\vec{s}_K) 
\frac{1}{\sqrt{|C_2|}} \sum_{\begin{array}{l} \scriptstyle \vec{u}_1 \in \mathbf{F}_q^{\dim P_J(C_2 \cap \ker(P_{\overline{J}}))}\\\scriptstyle \vec{u}_2 \in g_5(\vec{s}_W) \end{array}} |
\vec{s}_W, \vec{u}_1, g_6(\vec{s}_W, \vec{s}_K, \vec{u}_2),
\vec{0}, \vec{u}_2\rangle\right),
\]
and the reconstructed secret is not entangled with the
rest of quantum system.

Observe also that the number of qudits in the reconstructed part
is $\dim V = \dim \widetilde{P}_J(\ker(\widetilde{P}_{\overline{J}}))$ 
and if (\ref{eq402}) holds then
the entire secret is reconstructed.
Because 
the complement of any qualified set is forbidden
by \cite[Proposition 3]{ogawa05},
we see that the set $\overline{J}$ of participants
 has no information on the quantum secret
(\ref{eq3}) if (\ref{eq402}) holds.
This finishes the proof of Item \ref{l:decode} in Theorem \ref{thm1}.
\qed

\begin{example}\label{ex3}
We retain the notations from Example \ref{ex2}.
We have $J=\{1,2,3\}$, $\dim P_J(C_1) = 3$,
and $\dim P_J(C_2) = 1$.
$\dim P_J(C_1)/ P_J(C_2) = 2$.

When we express
\[
\vec{s} = \underbrace{a (2, 1, 0)}_{=\vec{s}_J} +
\underbrace{s_3 (0,0,1)}_{=\vec{s}_K} +
\underbrace{s_1 (1,0,0)}_{=\vec{s}_W},
\]
and fix $r$ in (\ref{eq:ex1}),
the index vector $\vec{x}$ in (\ref{eq:ex1}) becomes
\begin{eqnarray*}
\vec{x} &=& ( r+a + 3s_1 + 6s_3, r+ 5s_1 + 6s_3, 
r+6a + 6s_1 + 6 s_3 , \\ && \qquad r+3a + s_1 + s_3, r+3a + 4s_1 + s_3).
\end{eqnarray*}
$g_1((x_1, x_2, x_3)+P_J(C_2)) = (3x_2 - x_1 - 2x_3$,
$2x_2 - x_1 - x_3 ) =(a$, $s_1)$.
We have $C_2 \cap \ker(P_{\overline{J}}) = \{0\}$ and
$g_2$ is the zero map.
We have $g_3(x_1, x_2) = 2x_1 - x_3 = r+3a + 6s_3$
and $g_4(x_1, x_2) = (a,s_1,r+3a+6s_3)$.
Therefore,
after applying the proposed decoding procedure,
the state (\ref{eq:ex1}) of encoded shares becomes
\begin{eqnarray*}
&&\frac{1}{\sqrt{7}}\sum_{r \in \mathbf{F}_7}
\ket{a,s_1,r+3a+6s_3, r+3a + s_1 + s_3, r+3a + 4s_1 + s_3}\\
&=& \frac{1}{\sqrt{7}} \sum_{r' \in \mathbf{F}_7}
\ket{a,s_1,r'+6s_3, r' + s_1 + s_3, r' + 4s_1 + s_3}
\end{eqnarray*}
where $r' = r+3a$.

We see that
$s_1$ determines, independently of both $a$  and $s_3$, the coset
$\{ (r' + s_1 + s_3, r' + 4s_1 + s_3) \mid r' \in \mathbf{F}_7 \}$,
which is $g_5(\vec{s}_W)$.
$P_{\overline{J}}(\vec{x}) = (r' + s_1 + s_3, r' + 4s_1 + s_3)$, $s_1$
and $s_3$
uniquely determine $g_3(x_1, x_2,x_3)  = r'+6s_3$ which is $g_6$.
\end{example}

\section{Holevo Information and Coherent Information
of a Set of Shares}\label{sec11}
\subsection{Holevo Information}
In this section we prove that  both (\ref{eq101}) and (\ref{eq102})
are necessary for $J$ to be a qualified set.
We use the Holevo information \cite{chuangnielsen}
defined as follows.
Let $\mathcal{S}_{\mathrm{in}}$ and 
$\mathcal{S}_{\mathrm{out}}$ be sets of density matrices,
$\Gamma$ a completely positive trace-preserving
map from $\mathcal{S}_{\mathrm{in}}$ to
$\mathcal{S}_{\mathrm{out}}$,
$\{\rho_1$, \ldots, $\rho_m\} \subset \mathcal{S}_{\mathrm{in}}$,
and $P$ a probability distribution on $\{\rho_1$, \ldots, $\rho_m\}$.
The Holevo information is defined as
\begin{equation}
K(P, \{\rho_1, \ldots, \rho_m\}, \Gamma) =
H\left( \sum_{i=1}^m P(\rho_i) \Gamma(\rho_i) \right)
- \sum_{i=1}^m P(\rho_i) H(\Gamma(\rho_i)), \label{eq:holevo}
\end{equation}
where $H(\cdot)$ denotes the von Neumann entropy
counted in $\log_q$.
The Holevo information essentially
expresses the classical information that can be transferred
over $\Gamma$ \cite{chuangnielsen}.

Let $\Gamma_J$ be the completely positive trace-preserving map
from $\mathcal{S}(\bigotimes_{i=1}^{\dim C_1 - \dim C_2} \mathcal{G}_i)$
to $\mathcal{S}(\bigotimes_{j \in J} \mathcal{H}_j)$ induced by the
encoding procedure proposed in Section \ref{sec2},
where $\mathcal{S}(\cdot)$ denotes the set of density matrices
on a complex space $\cdot$. By $K_J$ we denote 
\begin{equation}
K(\textrm{uniform distribution}, \{ \ket{\vec{s}}\bra{\vec{s}}
\mid \vec{s} \in \mathrm{F}_q^{\dim C_1 - \dim C_2}\}, \Gamma_J).
\label{eq410}
\end{equation}

By \cite[Theorem 1]{ogawa05}
if
\begin{equation}
K_J < \dim C_1 - \dim C_2 \label{eq301}
\end{equation}
then
$J$ is not a qualified set.
The encoding procedure in Section \ref{sec2}
is a pure state scheme \cite[Section 2]{ogawa05},
that is,
the quantum state of all the shares is pure if
the encoded quantum secret is pure.
By \cite[Proposition 3]{ogawa05},
if $\overline{J}$ is not a forbidden set,
then $J$ is not a qualified set.
By \cite[Theorem 1]{ogawa05}
if
\begin{equation}
K_{\overline{J}} > 0 \label{eq302}
\end{equation}
then $\overline{J}$ is not a forbidden set.

We shall prove the next proposition.
By (\ref{eq101}), (\ref{eq102}), (\ref{eq301}) and (\ref{eq302}),
Proposition \ref{prop1} implies that both (\ref{eq101}) and (\ref{eq102}) are
necessary for $J$ to be a qualified set.
\begin{proposition}\label{prop1}
\begin{equation}
K_J = \dim P_J(C_1) - \dim P_J(C_2). \label{eq303}
\end{equation}
\end{proposition}
\begin{proof}
$\Gamma_J(\ket{\vec{s}}\bra{\vec{s}})$
is the partial trace of (\ref{eq4})
over $\bigotimes_{j\in \overline{J}} \mathcal{H}_j$.
By the definition of partial trace
\begin{eqnarray}
&&\Gamma_J(\ket{\vec{s}}\bra{\vec{s}})\nonumber\\
&=&
\frac{1}{|C_2|} 
\sum_{\vec{x}_1, \vec{x}_2 \in f(\vec{s})}
\ket{P_J(\vec{x_1})}\bra{P_J(\vec{x_2})}
\underbrace{ \langle P_{\overline{J}}(\vec{x}_1) | P_{\overline{J}}(\vec{x}_2)\rangle}_{= 1
\Leftrightarrow \vec{x}_2 \in \vec{x}_1 + \ker(P_{\overline{J}})} 
\nonumber\\
&=& \frac{1}{|C_2|} 
\sum_{\vec{u} \in P_{\overline{J}}(f(\vec{s}))}
\sum_{\vec{x}_1 \in  f(\vec{s}) \cap P_{\overline{J}}^{-1}(\vec{u})}
\sum_{\vec{x}_2 \in f(\vec{s}) \cap P_{\overline{J}}^{-1}(\vec{u})}
\ket{P_J(\vec{x_1})}\bra{P_J(\vec{x_2})}\nonumber\\
&=& 
\frac{1}{|C_2|} 
\sum_{\vec{u} \in P_{\overline{J}}(f(\vec{s}))}
\left(\sum_{\vec{x}_1 \in  f(\vec{s}) \cap P_{\overline{J}}^{-1}(\vec{u})}\ket{P_J(\vec{x_1})}\right)
\left(\sum_{\vec{x}_2 \in f(\vec{s}) \cap P_{\overline{J}}^{-1}(\vec{u})}
\bra{P_J(\vec{x_2})}\right)\nonumber\\
&=& 
\frac{1}{|C_2|} 
\sum_{\vec{u} \in P_{\overline{J}}(f(\vec{s}))}
\left(\sum_{\vec{x}_1 \in  f(\vec{s}) \cap ((\vec{0},\vec{u})+\ker(P_{\overline{J}}))}\ket{P_J(\vec{x_1})}\right)
\left(\sum_{\vec{x}_2 \in f(\vec{s}) \cap ((\vec{0},\vec{u})+\ker(P_{\overline{J}}))}
\bra{P_J(\vec{x_2})}\right) . \label{eq304}
\end{eqnarray}
For $\vec{u}_1$, $\vec{u}_2 \in P_{\overline{J}}(f(\vec{s}))$,
if $f(\vec{s}) \cap ((\vec{0},\vec{u}_1)+ \ker(P_{\overline{J}})) =f(\vec{s}) \cap ((\vec{0},\vec{u}_2)+ \ker(P_{\overline{J}}))$
then $\vec{x}_1$ and $\vec{x}_2$ in (\ref{eq304}) are taken over the same
set $P_J(\vec{x}) + P_J(C_2 \cap \ker(P_{\overline{J}}))$, where
$\vec{x}$ is any vector in $f(\vec{s}) \cap ((\vec{0},\vec{u}_1)+ \ker(P_{\overline{J}}))$.
Otherwise $\vec{x}_1$ and $\vec{x}_2$ in (\ref{eq304}) are taken over two
disjoint sets in $P_J(f(\vec{s}))$.
So (\ref{eq304}) is equal to
\begin{equation}
\frac{1}{|C_2|}
\sum_{A \in P_J(f(\vec{s}))/\sim} 
\left(\sum_{\vec{v} \in A} \ket{\vec{v}}\right)
\left(\sum_{\vec{v} \in A} \bra{\vec{v}}\right), \label{eq305}
\end{equation}
where $\sim$ is the equivalence relation that defines $\vec{v}_1$,
$\vec{v}_2 \in P_J(\mathbf{F}_q^n)$ to be equivalent
if $\vec{v}_1 \in
\vec{v}_2 + P_J(C_2 \cap \ker(P_{\overline{J}}))$.
(\ref{eq305}) is an equal mixture of $|P_J(C_2)/P_J(C_2\cap\ker(P_{\overline{J}}))|$
projection matrices to non-overlapping orthogonal
spaces,
therefore its von Neumann entropy is $\dim P_J(C_2) -
\dim P_J(C_2 \cap \ker(P_{\overline{J}}))$, 
which is the second term in the right hand side of (\ref{eq:holevo}).

By (\ref{eq305}),
the density matrix of the first term in RHS of of (\ref{eq:holevo})
is
\begin{eqnarray}
&&\frac{1}{q^{\dim C_1 - \dim C_2}} \sum_{\vec{s} \in \mathbf{F}_q^{\dim C_1-\dim C_2}}
\frac{1}{|C_2|}
 \sum_{A \in P_J(f(\vec{s}))/\sim} 
\left(\sum_{\vec{v} \in A} \ket{\vec{v}}\right)
\left(\sum_{\vec{v} \in A} \bra{\vec{v}}\right)\nonumber\\
&=&
\frac{1}{|C_1|},
 \sum_{A \in P_J(C_1)/P_J(C_2 \cap \ker(P_{\overline{J}}))} 
\left(\sum_{\vec{v} \in A} \ket{\vec{v}}\right)
\left(\sum_{\vec{v} \in A} \bra{\vec{v}}\right). \label{eq306}
\end{eqnarray}
The von Neumann entropy of (\ref{eq306})
is 
\begin{equation}
\dim P_J(C_1) -
\dim P_J(C_2 \cap \ker(P_{\overline{J}}))\label{eq411}
\end{equation}
 by the same argument as the last paragraph.
By (\ref{eq:holevo}) $K_J = \dim P_J(C_1) - \dim P_J(C_2)$.
\qed
\end{proof}

\subsection{Coherent Information}
We use the same notation as (\ref{eq:holevo}).
Denote by $\Gamma_E$ the channel to the environment
so that any pure state is mapped to a pure state
by $\Gamma \otimes \Gamma_E$.
The channel to the environment for $\Gamma_J$
is $\Gamma_{\overline{J}}$.
Then the coherent information of
the input state $\rho$ and the channel $\Gamma$
is defined by \cite{chuangnielsen}
\begin{equation}
H(\Gamma(\rho)) - H(\Gamma_E(\rho)). \label{eq:coherent}
\end{equation}
Equation (\ref{eq:coherent}) can become  negative.
The quantum capacity is expressed by the maximum
of the coherent information over $\rho$ \cite{devetak05}.

The coherent information of $\Gamma_J$ and
the completely mixed secret
$\frac{1}{q^{\dim C_1 - \dim C_2}}$ $\sum_{\vec{s}\in \mathbf{F}_q^{\dim C_1 - \dim C_2}}
\ket{\vec{s}}\bra{\vec{s}}$ is
(\ref{eq411}) subtracted by (\ref{eq411}) with $J$ substituted by
$\overline{J}$. Therefore the coherent information is
\begin{equation}
\dim P_J(C_1) - \dim C_2 \cap \ker(P_{\overline{J}})
- (\dim P_{\overline{J}}(C_1) - \dim C_2 \cap \ker(P_J)).
\label{cj}
\end{equation}
We consider to maximize (\ref{cj}) by replacing
$C_1$ by $D$ such that $C_2 \subset D \subset C_1$.
This amounts to maximize (\ref{eq:coherent})
over the quantum state completely mixed over
the subspace spanned by
$\{ \ket{\vec{s}} \mid
f(\vec{s}) \subset D\}$.

\begin{lemma}
Let $D$ be as above. Define
\[
D' = C_2 + (D \cap \ker(P_{\overline{J}})).
\]
Then we have
\begin{eqnarray}
&&\dim P_J(D) - \dim C_2 \cap \ker(P_{\overline{J}})
- (\dim P_{\overline{J}}(D) - \dim C_2 \cap \ker(P_J))\nonumber\\
&=& \dim P_J(D') - \dim C_2 \cap \ker(P_{\overline{J}})
- (\dim P_{\overline{J}}(D') - \dim C_2 \cap \ker(P_J)). \label{target}
\end{eqnarray}
\end{lemma}
\begin{proof}
Let $D = D' \oplus D''$. Then
$\dim D'' = \dim P_{\overline{J}}(D'')$ because
$D'' \cap \ker(P_{\overline{J}}) = \{\vec{0}\}$.
Therefore the $D''$ component in $D$ does not help to increase
the value of (\ref{cj}). Thus $D'$ yields the same value
for (\ref{cj}) as $D$ and we have (\ref{target}).
\qed
\end{proof}

So we see that $D = C_2 + (C_1 \cap \ker(P_{\overline{J}}))$
maximizes the coherent information to
its maximum value 
\begin{eqnarray*}
&& \dim P_J(C_2 + (C_1 \cap \ker(P_{\overline{J}}))) - \dim C_2 \cap \ker(P_{\overline{J}})\\
&&\mbox{ }- (\underbrace{\dim P_{\overline{J}}(C_2 + (C_1 \cap \ker(P_{\overline{J}}))}_{= \dim P_{\overline{J}}(C_2)} - \dim C_2 \cap \ker(P_J))\\
&=& \dim P_J(C_2 + (C_1 \cap \ker(P_{\overline{J}})))
- \underbrace{(\dim C_2 \cap \ker(P_{\overline{J}}) +
 \dim P_{\overline{J}}(C_2) - \dim C_2 \cap \ker(P_J))}_{= \dim P_J(C_2)}\\
&=& \dim \widetilde{P}_J(\ker \widetilde{P}_{\overline{J}}).
\end{eqnarray*}
We remark that  the proposed decoding procedure
in Section \ref{sec3} reconstructs precisely that number of qudits
in the secret.

\section{Analysis of the Conventional Scheme}\label{sec4}
In this section we show that the conventional
quantum ramp secret SS 
\cite{ogawa05} can be regarded as a special
case of the proposed construction,
and its qualified and forbidden sets can be identified by
Theorem \ref{thm1}.
Let $\alpha_1$, \ldots, $\alpha_n$ be pairwise distinct
nonzero\footnote{In \cite{ogawa05} $\alpha_i=0$ was not
explicitly prohibited,
but an author of \cite{ogawa05} informed that $\alpha_i$
must be nonzero for all $i=1$, \ldots, $n$.} elements
in $\mathbf{F}_q$, which correspond to $x_1$, \ldots, $x_n$
in \cite{ogawa05}. Denote $(\alpha_1$, \ldots, $\alpha_n)$ by
$\vec{\alpha}$.
Let $\vec{v} \in (\mathbf{F}_q\setminus \{0\})^n$.
Then the generalized Reed-Solomon code 
$\mathrm{GRS}_{n,k}(\vec{\alpha}$, $\vec{v})$ is \cite[Section 10.\S8]{macwilliams77}
\begin{equation}
\{ (v_1 h(\alpha_1), \ldots, v_nh(\alpha_n)) \mid \deg h(x) \leq k-1 \},
\label{eq8}
\end{equation}
where $h(x)$ is a univariate polynomial over $\mathbf{F}_q$.
Let $\vec{1} = (1$, \ldots, $1) \in \mathbf{F}_q^n$ and
$\vec{\alpha}^L = (\alpha_1^L$, \ldots, $\alpha_n^L) \in \mathbf{F}_q^n$.
The conventional scheme \cite{ogawa05}
is a special case of the proposed construction
with $C_1 =\mathrm{GRS}_{n,k}(\vec{\alpha}$, $\vec{1})$ and
$C_2 =\mathrm{GRS}_{n,k-L}(\vec{\alpha}$, $\vec{\alpha}^L)$.
Observe that $C_2 \subsetneq C_1$,
$\dim C_1 = k$, and $\dim C_2 = k-L$.
By the property of the generalized Reed-Solomon codes
(see e.g.\ \cite[Section 11.\S4]{macwilliams77}),
any subset $J \subseteq
\{1$, \ldots, $n\}$ satisfies both (\ref{eq101}) and (\ref{eq102})
if $|J| \geq \dim C_1$ and $|\overline{J}| \leq \dim C_2$.
Observe that the original restriction
$n =  \dim C_1 + \dim C_2$ \cite{ogawa05} is removed here.

\section{Algebraic Geometric Construction}\label{sec5}
In this section we give a construction of $C_1 \supset C_2$
based on algebraic geometry (AG) codes.
A major benefit of the AG codes is that $n$ can become arbitrarily
large for a fixed $q$ \cite{bn:stichtenoth}.
For terminology and mathematical notions of AG codes,
please refer to \cite{bn:stichtenoth}.
Let $F/\mathbf{F}_q$ be an algebraic function field
of one variable over $\mathbf{F}_q$,
$P_1$, \ldots, $P_n$ pairwise distinct places of degree one
in $F$, and $G_1$, $G_2$ divisors of $F$ whose supports
contain none of $P_1$, \ldots, $P_n$.
We assume $G_1 \geq G_2$.
Denote by $\mathcal{L}(G_1)$ the $\mathbf{F}_q$-linear
space associated with $G_1$.
The functional AG code associated with $G_1$, $P_1$, \ldots, $P_n$
is defined as
\[
C(G_1,P_1, \ldots, P_n) = \{ (f(P_1), \ldots, f(P_n)) \mid f \in \mathcal{L}(G_1) \}.
\]
Since $G_1 \geq G_2$ we have $C(G_1$, $P_1$, \ldots, $P_n) \supseteq C(G_2$,
$P_1$, \ldots, $P_n)$.
We further assume $C(G_1$, $P_1$, \ldots, $P_n) \neq C(G_2$,
$P_1$, \ldots, $P_n)$.
\begin{theorem}\label{thm2}
The ramp quantum SS constructed from
$C(G_1$, $P_1$, \ldots, $P_n) \supsetneq C(G_2$,
$P_1$, \ldots, $P_n)$ encodes $\dim C(G_1$,
$P_1$, \ldots, $P_n)- \dim C(G_2$,
$P_1$, \ldots, $P_n)$
qudits to $n$ shares. We have
\begin{eqnarray}
&&\dim C(G_1,P_1, \ldots, P_n)- \dim C(G_2,P_1, \ldots, P_n)\nonumber\\
& \geq& \deg G_1 - \deg G_2 - g(F), \label{eq201}
\end{eqnarray}
where $g(F)$ denotes the genus of $F$.
A set $J \subseteq \{1$, \ldots, $n\}$ is
a qualified set and its complement $\overline{J}$ is a forbidden set
if 
\begin{equation}
|J| \geq \max\{1+\deg G_1, n- (\deg G_2 - 2g(F)+1)\}. \label{eq502}
\end{equation}
\end{theorem}
\begin{proof}
Equation (\ref{eq201}) follows just from 
\begin{equation}
\dim C(G_1,P_1, \ldots, P_n) = \dim \mathcal{L}(G_1)
- \dim \mathcal{L}(G_1-P_1- \cdots - P_n), \label{eq203}
\end{equation}
and the Riemann-Roch theorem \cite{bn:stichtenoth}
\begin{equation}
\deg G_1 - g(F) + 1 \leq \dim \mathcal{L}(G_1) \leq \max\{0,\deg G_1 + 1\}, \label{eq202}
\end{equation}
where the left inequality of (\ref{eq202}) becomes equality
if
\begin{equation}
\deg G_1 \geq 2g(F) - 1.\label{eq210}
\end{equation}

Firstly we claim that (\ref{eq101}) and (\ref{eq102}) hold if
\begin{eqnarray}
|J| &\geq& 1+\deg G_1, \label{eq204}\\
|\overline{J}| &\leq& \deg G_2 - 2 g(F) + 1. \label{eq205}
\end{eqnarray}
By reordering indices we may assume that
$J = \{1$, \ldots, $|J|\}$.
Observe that
\begin{equation}
P_J(C(G_1,P_1, \ldots, P_n)) = C(G_1,P_1, \ldots, P_{|J|}). \label{eq206}
\end{equation}
If (\ref{eq204}) holds
then by (\ref{eq202}) we have
$\mathcal{L}(G_1 - P_1 - \cdots - P_{|J|}) = \{0\}$,
which means that $\mathcal{L}(G_1)$ is isomorphic to
$C(G_1$, $P_1$, \ldots, $P_{|J|})$ as an $\mathbf{F}_q$-linear
space by (\ref{eq203}).
By the same argument we also see that 
$\mathcal{L}(G_1)$ is isomorphic to
$C(G_1$, $P_1$, \ldots, $P_{n})$.
Thus we have seen that (\ref{eq204}) implies (\ref{eq101}).

If (\ref{eq205}) holds then
\[
\deg (G_2 - P_{|J|+1} - \cdots - P_n) \geq 2 g(F)-1,
\]
which implies by (\ref{eq210})
\begin{equation}
\dim \mathcal{L}(G_2 - P_{|J|+1} - \cdots - P_n) = \deg G_2 - |\overline{J}| - g(F) + 1. \label{eq211}
\end{equation}
By the same argument
\begin{equation}
\dim \mathcal{L}(G_2) = \deg G_2 - g(F) + 1. \label{eq212}
\end{equation}
Equations (\ref{eq203}), (\ref{eq211}) and (\ref{eq212})
imply $\dim C(G_2$, $P_{|J|+1}$, \ldots, $P_n) = |\overline{J}|$,
which in turn implies $C(G_2$, $P_{|J|+1}$, \ldots, $P_n) = \mathbf{F}_q^{|\overline{J}|}$.
Therefore we see that (\ref{eq205}) implies (\ref{eq102}).

Finally noting (\ref{eq502}) $\Rightarrow$ (\ref{eq204}) and (\ref{eq205})
finishes the proof. \qed
\end{proof}

\begin{remark}
As the generalized Reed-Solomon codes is a special case of
AG codes with $g(F)=0$ \cite{bn:stichtenoth},
Section \ref{sec4} can also be deduced from Theorem \ref{thm2}
instead of using \cite[Section 11.\S4]{macwilliams77}.
\end{remark}

\begin{theorem}
We retain notations from Theorem \ref{thm2} and
assume $\deg G_1 < n$.
The number (\ref{eq401}) of qudits in quantum secret that can be decoded
by $J$ is
\begin{equation}
\dim \frac{\mathcal{L}(G_1 - \sum_{j\in \overline{J}} P_j)+\mathcal{L}(G_2)}{(\mathcal{L}(G_1 - \sum_{j\in \overline{J}} P_j)+\mathcal{L}(G_2))
\cap (\mathcal{L}(G_1 - \sum_{j\in J} P_j)+\mathcal{L}(G_2))}.
\label{eq:thm3}
\end{equation}
\end{theorem}
\begin{proof}
Equation (\ref{eq401}) is equal to
\begin{equation}
\dim \ker(\widetilde{P}_{\overline{J}}) - \dim \ker(\widetilde{P}_J) \cap \ker(\widetilde{P}_{\overline{J}}). \label{eq2001}
\end{equation}
Since we assume $\deg G_1 < n$, the evaluation map
$h \in \mathcal{L}(G_1) \mapsto
(h(P_1)$, \ldots, $h(P_n) \in \mathbf{F}_q^n$ is injective and
we can deal with the space of functions
in $\mathcal{L}(G_1)$ to count the dimensions of (\ref{eq2001}).

For $h_1 + \mathcal{L}(G_2) \in \mathcal{L}(G_1)/\mathcal{L}(G_2)$,
its corresponding coset belongs to 
$\ker(\widetilde{P}_{\overline{J}})$ if and only if
there exists $h_2 \in \mathcal{L}(G_2)$ such that
$h_1(P_j) - h_2(P_j) = 0$ for all $j \in \overline{J}$,
which is equivalent to $h_1 - h_2 \in
\mathcal{L}(G_1 - \sum_{j\in \overline{J}}P_j)$.
In other words, the coset $h_1 + \mathcal{L}(G_2)$
satisfies the above condition if and only
if there exists $h'_1 \in \mathcal{L}(G_1 - \sum_{j\in \overline{J}}P_j)$
such that $h_1 \equiv h'_1 \pmod{\mathcal{L}(G_2)}$.
The dimension of space of cosets $h_1 + \mathcal{L}(G_2)$ 
with the above condition is
given by
\begin{equation}
\dim \frac{\mathcal{L}(G_1 - \sum_{j\in \overline{J}} P_j)+\mathcal{L}(G_2)}{\mathcal{L}(G_2)}. \label{eq2002}
\end{equation}

Moreover, while satisfying the condition of the last paragraph,
the coset corresponding to $h_1 + \mathcal{L}(G_2)$
belongs to $\ker(\widetilde{P}_J)$ if and only if
there exists another
$h^{\prime\prime}_1 \in \mathcal{L}(G_1 - \sum_{j\in J}P_j)$
such that $h_1 \equiv h^{\prime\prime}_1 \pmod{\mathcal{L}(G_2)}$.
The dimension of space of cosets $h_1 + \mathcal{L}(G_2)$ 
with the above two conditions is
given by
\begin{equation}
\dim \frac{(\mathcal{L}(G_1 - \sum_{j\in \overline{J}} P_j)+\mathcal{L}(G_2))
\cap (\mathcal{L}(G_1 - \sum_{j\in J} P_j)+\mathcal{L}(G_2))}{\mathcal{L}(G_2)}. \label{eq2003}
\end{equation}
By (\ref{eq2001}), subtracting (\ref{eq2003}) from (\ref{eq2002}) gives
(\ref{eq:thm3}).
\qed
\end{proof}

\section{Conclusion}\label{sec6}
We have shown that
a quantum ramp secret sharing scheme can be constructed from
any nested pair of linear codes,
and also shown necessary and sufficient conditions for
the qualified and the forbidden sets as Theorem \ref{thm1}.
A construction of nested linear codes is given by
the algebraic geometry in Theorem \ref{thm2}.
The following issues are future research agenda.

What is a better construction of  $C_1 \supsetneq C_2$ 
than Theorem \ref{thm2} when $q < n$?
In particular, (\ref{eq205}) should use
both divisors $G_1$ and $G_2$ because 
(\ref{eq101}) and (\ref{eq102}) use both of nested linear codes.
Also, $J$ corresponds to a set of $\mathbf{F}_q$-rational points
on an algebraic curve
when AG codes are used, but only the size of $J$ is taken into
account in (\ref{eq205}).
The geometry of $J$ should also be taken into account.
We shall investigate them in future.

\begin{acknowledgements}
The author would like to thank Profs.\ Ivan Damg\aa rd,
Johan Hansen, Olav Geil, Diego Ruano, and Dr.\  Ignacio Cascudo,
for helpful discussions.
He would also like to thank Prof.\ Tomohiro Ogawa
for clarification of \cite{ogawa05}.
This research is partly supported by the National
Institute of Information and Communications Technology,
Japan, by the Japan Society for the Promotion of Science Grant  
 Nos.\ 23246071 and 26289116,
and the                                                                   
 Villum Foundation through their VELUX Visiting Professor Programme             
 2013--2014.
\end{acknowledgements}


\end{document}